%% file: main.tex
\documentclass[]{article}
\usepackage{ltexpprtMAIMED}
\usepackage{sodafixes}

\usepackage{microtype}
\usepackage{amsmath}
\usepackage{amssymb}
\usepackage{hyperref}
\usepackage{url}
\usepackage{enumitem}
\usepackage{flushend}

\usepackage{graphicx}
\graphicspath{{illustrations/}}
\newtheorem{definition}{Definition}

\newcommand{\BO}[1]{{ O}\left(#1\right)}

\newcommand{\SO}[1]{{o}\left(#1\right)}
\newcommand{\BT}[1]{{\Theta}\left(#1\right)}
\newcommand{\BOM}[1]{\Omega\left(#1\right)}

\newcommand{\E}{\mathop{\text{E}\/}}

\newcommand{\modl}{\mathop{\text{ mod$_*$ }\/}}    

\newcommand{\B}[1]{\mathcal{B}(#1)}
\renewcommand{\b}[1]{B(#1)}

\title{Distance Sensitive Bloom Filters Without False Negatives\thanks{The research leading to these results has received funding from the European Research Council under the EU 7th Framework Programme, ERC grant agreement no. 614331.}}

\author{
Mayank Goswami\thanks{Queens College, CUNY, mayank.goswami@qc.cuny.edu. Part of this work was done at the Max-Planck Institute for Informatics, Germany.}
\and
Rasmus Pagh\thanks{IT University of Copenhagen, Denmark, 
  pagh@itu.dk.}
\and
Francesco Silvestri\thanks{IT University of Copenhagen, Denmark, 
  fras@itu.dk.}
\and
Johan Sivertsen\thanks{IT University of Copenhagen, Denmark, 
  jovt@itu.dk.}
}

\date{}

\begin{document}

\maketitle

\begin{abstract}
  A Bloom filter is a widely used data-structure for representing a set $S$ and answering queries of the form ``Is $x$ in $S$?''.
  By allowing some false positive answers (saying `yes' when the answer is in fact `no') Bloom filters use space significantly below what is required for storing $S$.
  In the \emph{distance sensitive} setting we work with a set $S$ of (Hamming) vectors and seek a data structure that offers a similar trade-off, but answers queries of the form ``Is $x$ \emph{close} to an element of $S$?'' (in Hamming distance).
  Previous work on distance sensitive Bloom filters have accepted false positive \emph{and} false negative answers.
  Absence of false negatives is of critical importance in many applications of Bloom filters, so it is natural to ask if this can be also achieved in the distance sensitive setting.
  Our main contributions are upper and lower bounds (that are tight in several cases) for space usage in the distance sensitive setting where false negatives are not allowed. 
\end{abstract}

\section{Introduction}
\input{introduction}

\section{Problem definition and notation}
\input{prelim}

\section{Lower bounds}
\input{ac_lower}
\input{wc_lower}

\section{Upper bounds}
\label{sec:upper-bounds}
\input{wc_upper}

\section{Conclusion}
\input{conclusion}

\section*{Acknowledgements}
The authors would like to thank Thomas Dybdahl Ahle for many fruitful discussions on bounding the size of the Hamming ball.

\bibliographystyle{plainurl}
\bibliography{biblio}
\vspace{0.1em}

\end{document}

%% file: introduction.tex
The Bloom filter~\cite{Bloom1970} is a well-known data structure for answering \emph{approximate membership queries} on a set $S$, i.e., queries of the form ``Is $x$ in $S$?''.
Bloom filters are widely used in practice because they require less space than a dictionary data structure for storing $S$. This is achieved by allowing a certain probability of \emph{false positives}, i.e., `yes' answers for queries $x\not\in S$. It is critical for many applications of Bloom filters that errors are one-sided, i.e., `no' answers are always correct. In other words, \emph{false negatives} do not occur.

Generally the set $S$ that we want to ask questions about is a subset from some much larger domain.
In applications of Bloom filters the answer to a membership query should most often be negative, and for the vast majority of such queries the Bloom filter will give the correct answer.
Whenever the filter does give a positive answer, correctness can often be  checked using a slower, less space-efficient method (even possibly on a different machine).
Bloom filters are often used as part of an exact two-level data structure where it acts as the first level that is cheap to use and does most of the work, and the second level is more expensive but only rarely needed.   
Having false negatives means this setup fails, and the application would have to either accept some possibility of getting a wrong answer or perform an expensive exact query every time.

In this paper we present upper and lower bounds on the space complexity of filters for \emph{distance sensitive approximate membership queries}.
These filters answer queries of the form ``Is $x$ close to some element of $S$?'' 
Specifically, we address this question in the $d$-dimensional Hamming space where $x\in \{0,1\}^d$, $S\subset \{0,1\}^d$ is a set of $n$ points, and ``close'' means within a given Hamming distance~$r\in\{1,\cdots,r\}$.
In contrast to previous work on this problem, the filters presented in this paper introduce no false negatives.

We study distance sensitive filters under an approximation factor $c\geq1$: a small false positive rate is required for points at distance more than $cr$ from the query point, while  no rate guarantee is required for points at distances between $r$ and $cr$.
This kind of approximation of distances is standard in data structures for high-dimensional search.

\subsection{Motivation}
There are many potential applications for this kind of data structure.
As a concrete example, consider a journal comprising a large collection of academic papers. When accepting a new paper the journal might want to check if the new paper is very similar to any prior work already published. By using a distance-sensitive  filter this can be done in a space-efficient manner.
Because we do not allow false negatives, any new paper passing this test (with a `no' result) is guaranteed to be significantly different from all prior work. In the rare case that a paper fails the test, the submission process could be halted pending a consultation of the full archive.
Furthermore, since the filter provides very little information about the content of the papers it would not need to be subject to the same access control as a full database of all the journals papers might be under.
More interesting examples of applications for distance-sensitive  filters can be found in~\cite{Kirsch} and for Bloom filters in general in~\cite{broder2004network}.

\subsection{Our results}

We study the space required for answering distance-sensitive approximate membership queries with no false negatives.
It turns out that, in contrast to approximate membership, we get different bounds depending on how the false positive rate is defined:
\begin{itemize}
\item If we desire a \emph{point-wise} error bound (Definition~\ref{def:DSAM}) for each query at distance $\geq cr$ from $S$, the space usage must be $\BOM{n \left( \frac{r^2}{d} + \log \frac{1}{\varepsilon}\right)}$ for almost all parameters, and $\BOM{n \left(\frac{r}{c}+\frac{c}{c-1} \log \frac{1}{\varepsilon}\right)}$ bits if $n$ is not too large
(see Theorem~\ref{thm:wclb1}).
\item If it suffices to have an $\varepsilon$ \emph{average} false positive rate (Definition~\ref{def:EDSAM}) over all queries at distance $\geq cr$ from $S$, where $Cl<d/2$, the space usage must be $\BOM{n \left( \frac{r^2}{d} + \log \frac{1}{\varepsilon}\right)}$ bits.
(see Theorem~\ref{avg_error_thm}).
\end{itemize}
We match these lower bounds with almost tight upper bounds on space usage in Section.~\ref{sec:upper-bounds}.
We introduce the notion of vector \emph{signature}, which can be seen as a succinct version of a {\sc CountSketch}~\cite{CharikarCF04}, and then show how to use them to design distance sensitive filters with point-wise and average errors.

Our focus is on space usage rather than query-time, and indeed it would be surprising if poly-logarithmic query time in $n$ were possible since our (point-wise) filter could be used, say with $\varepsilon = 1/n$, to solve the $c$-approximate nearest neighbor problem. The best currently know data structures for this problem use $n^{\Omega(1/c)}$ time~\cite{andoni2015optimal}.

\subsection{Related work}

There is little prior work specifically on distance sensitive approximate membership.
The problem corresponds to querying a standard Bloom filter in a ball around the query point, but this solution is slow, time $\Omega(\binom{d}{r})$, and also not particularly space efficient since we would need to use a Bloom filter with a very small false positive rate to bound the probability that none of the queries yield a false positive.
More precisely, the required space usage for this approach would be $\Omega(n r \log\frac{d}{r})$ bits~\cite{Carter1978}.

Mitzenmacher and Kirsch~\cite{Kirsch} considered data structures that look like Bloom filters but replace standard hash functions with locality sensitive hash (LSH) functions~\cite{Indyk1998} to achieve distance sensitivity.
However, this approach introduces false negatives because LSH is not guaranteed to produce collisions.
In order to reduce the number of false negatives the conjunction used when querying Bloom filters is replaced by a threshold function: There should just be ``many'' hash collisions.
Unfortunately, the achieved approximation factor is large, $c=\BO{\log n}$.
Hua et al.~\cite{Hua2012} extended the data structure of~\cite{Kirsch} with practical improvements and provided extensive experiments, confirming that false negatives also appear in practice.

There has been some recent progress on developing LSH families that can answer near neighbor queries without false negatives~\cite{Pagh2015}, but it seems inherent to such families that the storage cost grows exponentially with~$r$. Thus this approach is not promising perhaps except for very small values of $r$.

Finally,  it is known that allowing a constant fraction of false negatives does asymptotically  affect the space usage that can be achieved by approximate membership data structures~\cite{pagh2001lossy},
so it is not apriori clear that space bounds will be worse than when false negatives are not allowed.


%% file: prelim.tex
The Hamming distance $D(p,q)$ between two points $p,q\in\{0,1\}^d$ is the number of positions where $p$ and $q$ differ. 
Given a set $S\subseteq \{0,1\}^d$ of $n$ points and a point $q\in \{0,1\}^d$, 
we overload the meaning of $D(\cdot)$ by defining $D(q,S)$ to be the minimum distance between $q$ and any point in $S$, i.e. $D(q,S)=\min_{p\in S}D(q,p)$.
We use $\binom{A}{n}$ to denote $\{S\subseteq A:|S|=n\}$ when $A$ is a set.
We let $\B{q,r,d}$ be the $d$-dimensional Hamming ball of radius $r$ centered around $q$, that is  $\B{q,r,d}=\{p\in \{0,1\}^d, \, D(p,q)\leq r\}$, and we let $\b{r,d}$ denote its size, being its size independent of $q$.

We formally define \emph{distance sensitive approximate membership filters}  as follows:

\begin{definition}
  \label{def:DSAMmother}
  {\sc (Distance sensitive  approximate membership filter)}
Let $r\geq 0$, $c\geq 1$, and $\varepsilon\in[0,1]$. 
Given a set $S \subset \{0,1\}^d$ define the two sets:
\begin{align*}
  Q_\text{near}&=\{x\in\{0,1\}^d:D(x,S)\leq r\},\\
  Q_\text{far}&=\{x\in\{0,1\}^d:D(x,S)> cr\}.
\end{align*}

A $(r,c,\varepsilon)$-distance sensitive approximate membership filter for $S$
is a data-structure that on a query $q\in\{0,1\}^d$ reports: 
\begin{itemize}
\item `\emph{Yes}' if $q\in Q_\text{near}$
\item `\emph{No}' if $q\in Q_\text{far}$, but with some probability of error (i.e. false positives).
\end{itemize}
If $q\notin Q_\text{near}\cup Q_\text{far}$ the data structure can return any answer.
\end{definition}

\begin{figure}[ht]
  \centering
\includegraphics{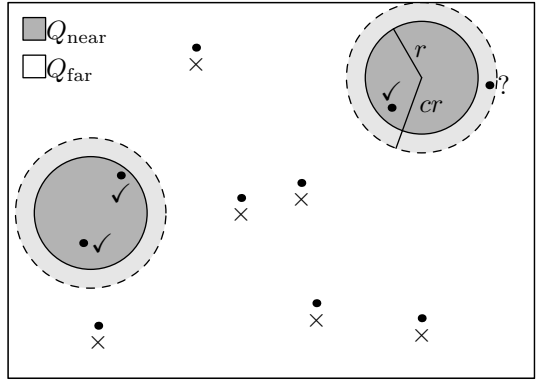}
  \caption[DAMQ data-structure]{Illustration for $n=2$ showing some queries with their desired output: $\checkmark \rightarrow$ `\emph{Yes}', $\times\rightarrow$`\emph{No}', ? $\rightarrow$ Undefined.}
  \label{fig:damq}
\end{figure}

In the rest of the paper, we study space bounds under two error measures, named point-wise and average errors.

\begin{definition}[Point-wise error]
  \label{def:DSAM}
A $(r,c,\varepsilon)$-distance sensitive approximate membership filter for $S$ has point-wise error $\varepsilon $ if, on a query $q\in\{0,1\}^d$, it reports: 
\begin{itemize}
\item `\emph{Yes}' if $q\in Q_\text{near}$;
\item `\emph{No}' with probability at least $1-\varepsilon$ if $q\in Q_\text{far}$ (the probability is over the random choices of the filter).
\end{itemize}
\end{definition}

This is a strong guarantee since each  point in $Q_\text{far}$ has probability $\varepsilon$ to fail.
If hard queries are not expected it might be acceptable that some points give false positives in every instance of the data structure, as long as only an $\varepsilon$ total fraction of points in $Q_\text{far}$ give false positives.
We refer to this weaker filter as the \emph{average error} version:

\begin{definition}[Average error]
\label{def:EDSAM}
A $(r,c,\varepsilon)$-distance sensitive approximate membership filter for $S$ has average error $\varepsilon $  if, on a query $q\in\{0,1\}^d$, it reports: 
\begin{itemize}
\item `\emph{Yes}' if $q\in Q_\text{near}$;
\item `\emph{No}' with probability at least $1-\varepsilon$, if $q$ is randomly and uniformly selected from $Q_\text{far}$ (the probability is over the random selection in $Q_\text{far}$ and over the random choices of the filter).
\end{itemize}
\end{definition}

The average-error guarantee implies that the filter 
provides the correct answer to at least a $(1-\varepsilon)$ fraction, in expectation,  of the points in $Q_\text{far}$.
Clearly, a filter with point-wise error is also a filter with average error.
Though the difference between these two error measures may seem small, their properties and analysis differ substantially.


%% file: ac_lower.tex
As a warm-up, we first investigate what can be done when no errors are allowed, that is when $\varepsilon=0$ (in this case the average and point-wise error guarantees are equivalent).
The next theorem shows that, up to constant factors, the optimal filter is no better than one that stores $S$ explicitly.
When $\varepsilon = 0$ there is no distinction between point-wise and average error. 
Throughout this paper we let $\log x$  denote the logarithm of $x$ in base 2.

\begin{theorem}\label{eps0}
Any distance sensitive approximate membership filter with error $\varepsilon = 0$ must use at least 
$$n \log \left(\frac{2^d}{e n \b{cr,d}}\right)$$
 bits in the worst case. If $d = \omega(\log n)$ and $cr = \SO{d/\log d}$ then it must use $\BOM{nd}$ bits. 
\end{theorem}
\begin{proof}
  The proof is an encoding argument. A set $S \subseteq \{0,1\}^d$ of size $n$ is encoded by Alice and sent to Bob who will recover it.
  Assume the optimal filter  uses $s$ bits in the worst case.
  Alice inserts the given set $S$ into the optimal filter, and runs the query algorithm on each point in the universe.
  Since there are no false positives, the filter says `yes' to a set $P$ of at most $n \b{cr,d}$ points.
  Alice encodes $S$ as a subset of $P$ using $\log  \binom{n  \b{cr,d}}{n}+\BO{1}$ bits.
  Alice then sends the at most $s$ bits of the optimal filter to  Bob along with the strings encoding $S$ as a subset of $P$.

  The decoding procedure is straightforward.
  Bob queries the optimal filter with all points in $\{0,1\}^d$, recovering $P$.
  Then, using $P$ and the second string of bits received from Alice, Bob can recover the initial set $S$.
  
  Since every set $S$ of size $n$ can be encoded, we get that:
\begin{eqnarray*}
 s +\log  \binom{n \b{cr,d}}{n} &\geq& \log  \binom{2^{d}}{n}  
\end{eqnarray*}  
from which follows that
\begin{eqnarray*}
  s & \geq & \log  \left( \left(\frac{2^{d}}{n}\right)^{n} / \left( \frac{en \b{cr,d}}{n}\right)^{n} \right)  \\
   & \geq & nd - n \log  (en) - n \log  \b{cr,d}
\end{eqnarray*}

If $d = \omega(\log n)$, we get $s = \BOM{nd - n \log \b{cr,d}}$.
Further, using that $\b{cr,d} = \sum_{i=0}^{cr} \binom{d}{i} < d^{cr}$ for $cr <d/2$, we get that $s = \BOM{nd - ncr \log d}$, which is $\BOM{nd}$ when $cr = \SO{d/\log d}$.
This establishes the theorem.
\end{proof}

\subsection{Average error}

Next we investigate the distance sensitive membership problem with average error $\varepsilon>0$.

\begin{theorem}\label{avg_error_thm}
  Assume that $ nB(cr,d)/2^d < \varepsilon < 1/4 $. Then any distance sensitive membership filter with average error $\varepsilon$ must use 
  $$\BOM{ n \left(\frac{r^2}{d} + \log \left(\frac{1}{\varepsilon}\right)\right)}$$
   bits in the worst case.
\end{theorem}

Before proving the theorem, we highlight some remarks:
\begin{enumerate}
\item The above theorem holds as long as $nB(cr,d)< 2^{d-2}$ (the union of the balls of radius $cr$ around points in the input is less than a quarter of the full Hamming space) and $nB(cr,d)/2^{d}<\varepsilon< 1/4$.
  This is the most interesting range of parameters. As we will see later, the $\Omega(nr^{2}/d)$ lower bound holds as long as $nB(cr,d)<2^{d-1}$, and it starts to deteriorate when $nB(cr,d)$ approaches $2^d$.
  It is clear that some upper bound on $nB(cr,d)$ is necessary; if it approaches size $2^{d}-O(n/d)$, then storing the complement exactly in $O(n)$ bits suffices. Also note that at the lower limit of $\varepsilon= nB(cr,d)/2^d$, this lower bound matches the lower bound of the $\varepsilon=0$ case in Theorem~\ref{eps0}. Thus Theorem~\ref{eps0} follows from Theorem~\ref{avg_error_thm}.

\item  The term $B(cr,d)$ has no simple closed expression for all $c$ and $r$, and so the dependence of the hypothesis of the theorem on $c$, $r$ and $d$ is not straightforward.

\end{enumerate}
The rest of this section is devoted to the proof of Theorem~\ref{avg_error_thm}.

\begin{proof}
The proof is derived for a deterministic version of the distance sensitive membership filter: in this setting, the filter answers `no' to at least a fraction of points in $Q_\text{far}$ (i.e., points at distance at least $cr$ from all points in the input point set $S$), and hence there can be at most $\varepsilon |Q_{\text{far}}|$ false positives. 
We claim that such a lower bound applies also to a randomized filter. 
Suppose that a randomized filter  requires  $s$ bits, with $s$ smaller than the lower bound.
Since the expected number of correct `no' answers is at least $(1-\varepsilon)|Q_\text{far}|$, there must exist random values for which the filter provides the correct solution for at least $(1-\varepsilon) | Q_{\text{far}}|$ points: by using these values, we obtain a deterministic average error filter with space complexity $s$ lower than the lower bound, which is a contradiction.

We first prove a $\BOM{n \log(1/\varepsilon)}$ lower bound. The proof is an encoding argument that extends the scheme presented in the proof of Theorem~\ref{eps0} and in \cite{Carter1978}.
Alice receives a set $S$ of size $n$ from the universe to encode. Assume the optimal distance sensitive filter with $\varepsilon$ average error uses $s$ bits in the worst case.
Alice inserts~$S$ into the filter, and runs the query algorithm on all points in the universe recovering $P$, the set of points the filter answers `Yes' to.
We first claim that $|P|\leq2^{d+1}\varepsilon$.
First, the number of positives not considered false is at most $nB(cr,d)$ (this bound is achieved when all the balls are disjoint), which is less than $2^{d}\varepsilon$. Also the number of false positives is always at most $2^d \varepsilon$.
Adding these, we find that the total number of positives is at most $2^{d+1}\varepsilon$.
Alice then encodes the set $S$ as a subset of $P$, using at most $\log \binom{2^{d+1}\varepsilon}{n}$ bits. Alice sends these bits to Bob along with the at most $s$ bits representing the optimal filter for $S$.

Bob queries the filter with all $q\in\{0,1\}^d$ and recovers $P$. Bob then uses the extra bits sent by Alice to find the subset of $P$ identical to $S$.
We have that:
\begin{eqnarray*}
s + \log \binom{2^{d+1}\varepsilon}{n} &\geq& \log \binom{2^d}{n}  \\  
\Rightarrow s &\geq& \log \frac{2^d \cdots (2^d-n+1)}{(\varepsilon 2^{d+1})\cdots(\varepsilon 2^{d+1}-n+1)}  \\
\Rightarrow s &\geq&   \log \left(\frac{2^d }{ \varepsilon 2^{d+1}}\right)^n  \\
\Rightarrow s &\geq&  n\log\left(\frac{1}{2\varepsilon}\right) \in \BOM{n \log\left(\frac{1}{\varepsilon}\right)}. 
\end{eqnarray*}

To prove the $n r^{2}/d$ lower bound, we first develop some notation.
Consider the hypercube graph on the $d$-dimensional Hamming cube where two points $p$ and~$q$ have an edge between them if they have Hamming distance $1$. Given a set $A \subset \{0,1\}^d$, let $A^c$ denote its complement, and define $\partial A$ to be the set of points in $A$ that have an edge to a point in $A^c$ (when either $A^c$ or $A$ is empty, $\partial A$ is the empty set). Also, given an integer $r> 0$, define $A^{-r} = A \setminus \bigcup_{x \in \partial A} \B{x,r-1,d}$. $A^{-r}$  contains exactly those points $x \in A$ such that the ball $\B{x,r,d}$ is contained inside $A$.

A deterministic filter that uses $s$ bits can be viewed as a function $\mathcal{F}: \binom{\{0,1\}^d}{n} \rightarrow \{0,1\}^s$; given a set $S \subseteq \{0,1\}^d$ of size $n$, $\mathcal{F}(S)$ is the memory representation of $S$ that uses at most $s$ bits. Let $V(S) = |\cup_{x \in S} \B{x,r,d}| + \varepsilon( 2^d - |\cup_{x \in S} \B{x,r,d}|)$: we note that $V(S)$ is an upper bound to the number of `yes' answers returned by the filter (i.e., both true and false positives), and  $V(S) \leq 2^{d-1}$ by the hypothesis of the theorem.

Running the query algorithm on all points in the Hamming cube for the representation $\mathcal{F}(S)$ returns a set~$P_{S}$ of positives (${P_{S}}^c$ of negatives) such that $|P_{S}| \leq V(S)$. Let us denote by $D$ the function that takes in a set $S$, and outputs the set $P_{S}$ of positives returned by the query algorithm on the representation $\mathcal{F}(S)$. 

Varying over all $S \in  \binom{\{0,1\}^d}{n}$, we get a family $\mathcal{T}$ of sets such that:
\begin{enumerate}
\item $\forall S$, $\exists P \in \mathcal{T}$ such that $\B{x,r,d} \subset P$ for all $x \in S$. 
\item For any $P \in \mathcal{T}$ and $\forall S$ such that $D(S) = P$, $|P| \leq V(S)$.
\end{enumerate}

Thus $D$ is a function from $\{0,1\}^s$ to $\mathcal{T}$, the image of which is all of $\mathcal{T}$. This implies that $s\geq \log |\mathcal{T}|$.
So in order to get a lower bound on $s$ it suffices to get a lower bound on the size of the smallest family $\mathcal{T}$ with the above properties.

Fix $P \in \mathcal{T}$. Define $D^{-1}(P)=\{S: D(S) = P\}$. Any ball of radius $r$ around a point $p \in S$ such that $S \in D^{-1}(P)$ must be completely contained inside $P$. The maximum number of such points $p$ is $|P^{-r}|$. Thus we get that $|\cup_{S \in D^{-1}(P)} S| \leq |P^{-r}|$. This implies that $|D^{-1}(P)| \leq \binom{|P^{-r}|}{n}$.

Since all possible sets (from $\binom{\{0,1\}^d}{n}$) need to be covered, we get that $|\mathcal{T}| \geq  \binom{2^d}{n} / \binom{|P^{-r}|}{n}$. We now need an upper bound on the size of $|P^{-r}|$.
Lemma~\ref{lem:ball} states that $|P^{-r}| \leq 2^d e^{-2r^2/d}$.


The proof of the lower bound in Theorem~\ref{avg_error_thm} then follows by applying Lemma~\ref{lem:ball}:
\begin{align*}
 |\mathcal{T}|  &\geq  \binom{2^d}{n} / \binom{|P^{-r}|}{n} \\
 & \geq  \left( \frac{e 2^d}{|P^{-r}|} \right)^{n} \\ 
 & \geq e^{n\left(2r^2/d+1\right)}  , 
\end{align*}
  which implies that
$ s \geq \log \mathcal{T} = \Omega \left(nr^{2}/d  \right)$.
Combining our bounds, we get that when $n,r$ and $c$ satisfy the condition that $nB(cr,d) \leq 2^{d-2}$, any filter must use $\Omega ( n (r^2/d + \log (1/\varepsilon)))$ bits in the worst case. 
\end{proof}

\begin{lemma}\label{lem:ball}
Let $S$, $P$ and $r$ be as above. Then $|P^{-r}| < 2^d e^{-2r^2/d}$. 
\end{lemma}

\begin{proof}
Note that $P$ is the set of positives (after running the query algorithm on all points in the Hamming space) on the filter $\mathcal{F}(S)$. Thus we have that $|P| \leq V(S) \leq 2^{d-1}$. The size of $P^{-r}$ increases as $P$ increases, so we have that $|P^{-r}|$ is at most $\max |A^{-r}|$, where the maximum is taken over all sets $A$ such that $|A| = 2^{d-1}$.

We will first prove that if $|A| = 2^{d-1}$, then $\max |A^{-r}|$ is at most $B(d/2-r,d)$ (the size of the Hamming ball of radius $d/2 - r$). The proof is by induction (the statement is actually true for any $r < d/2$, not just the input parameter $r$, and so we will treat it as a variable). 

For $r=1$, the statement is that $|A^{-1}|$ is maximized when $A$ is the Hamming ball of radius $d/2$. This is the statement of Harper's theorem, also called the vertex-isoperimetric inequality \cite{Bollobas:1986:CSS:7228}, that states that Hamming balls have the smallest vertex boundary among all sets of a given size.

Assume now that the statement is true for $r=k$, i.e., of all sets $A$ such that $|A| = 2^{d-1}$, the one that maximizes $|A^{-k}|$ is the Hamming ball of radius $d/2$. In this case, note that $A^{-k}$ is the Hamming ball of radius $d/2-k$.

Assume that the statement for $r=k+1$ is false, i.e., there is a set $W$ (of size $2^{d-1}$) such that $|\B{0,d/2,d}^{-(k+1)}| < |W^{-(k+1)}|$. Note that by the inductive hypothesis, we know that $|\B{0,d/2,d}^{-k}| \geq |W^{-k}|$.

However, the vertex-isoperimetric inequality can also be stated as: if a set $W$ (that is not a ball) has size greater then or equal to that of the Hamming ball of radius $R$, then $|W \cup \Gamma(W)|$ is larger than the size of Hamming ball of radius $R+1$, where $\Gamma(W)$ is the set of neighbors of $W$. Thus  $|\B{0,d/2,d}^{-(k+1)}| < |W^{-(k+1)}|$ actually implies $|\B{0,d/2,d}^{-k}| < |W^{-k}|$, which contradicts the inductive hypothesis.

Finally, we bound $B(d/2-r,d)$ using  the following Chernoff-Hoeffding bound \cite{mitzenmacher2005probability} for binomial random variables:

If $X_{i}$ denotes the outcome of the $i$th coin toss with an unbiased coin, and $X = \sum_{i=1}^{d} X_{i}$, then $\Pr[X \leq \mu - a] \leq e^{-2a^2/d}$, for all $0 < a < \mu$, where $\mu = \mathbb{E}[X] = d/2$. 
Let $X \sim \text{Bin}(d,0.5)$. Now we have that
\begin{align*}
  |P^{-r}| \leq &B(d/2-r,d) \\
  =& 2^d P[X \leq d/2-r]\\
\leq & 2^d e^{-2r^{2}/d}.
\end{align*}
\end{proof}


%% file: wc_lower.tex
\subsection{Point-wise error}
The lower bound for the average case in Theorem~\ref{avg_error_thm} also applies to a filter with point-wise error guarantees.
A $(r,c,\varepsilon)$-filter with point-wise error $\varepsilon $ is also a $(r,c,\varepsilon)$-filter with average error $\varepsilon$:
if each point fails with probability $\varepsilon$, then a random point fails with probability $\varepsilon$.
However, a stronger lower bound holds for point-wise error if the number of points $n$ is not too large. 

\begin{theorem}\label{thm:wclb1}
  Consider an $(r,c,\varepsilon)$-distance sensitive approximate membership filter with point-wise error  on a set $S$ of $n$ points in $\{0,1\}^d$.
  Then, in the worst case, the filter must use:
\begin{itemize}
\item  $\BOM{n \left(\tfrac{r^2}{d} + \log \tfrac{1}{\varepsilon}\right)}$ bits if $n \b{cr,d}/2^d<\varepsilon<1/4$.
\item  $\BOM{n \left(\tfrac{r}{c} + \log\tfrac{1}{\varepsilon} \right)}$ bits  if $n \b{cr,\delta cr}/2^{\delta cr} < \varepsilon<1/4$ for some constant $\delta$.
\end{itemize}
\end{theorem}
\begin{proof}
As already said, the first bound follows by applying Theorem~\ref{avg_error_thm} since a $(r,c,\varepsilon)$-filter with point-wise error  is also a $(r,c,\varepsilon)$-filter with average error.

We now prove the second claim.
Observe that a filter for $d$-dimensional points with point-wise error $\varepsilon $ is also a filter for $d'$-dimensional points with the same guarantees when $d>d'$.
Then, the lower bound obtained by  Theorem~\ref{avg_error_thm} for dimension $d'=\delta cr$, for some small constant $\delta$, applies to dimension $d$, and it is also stronger since the lower bound in  Theorem~\ref{avg_error_thm}  is {decreasing} in $d$.
However, the new bound needs to meet the condition of Theorem~\ref{avg_error_thm}:
given a filter for dimension  $d'=\delta cr$, then the condition states that $n\b{cr, \delta cr}/2^{\delta cr} < \varepsilon < 1/4$. The theorem follows.
\end{proof}

We observe that the proof used to derive the stronger lower bound does not work for the average error measure: indeed, the average error rate relatively to a subspace (e.g., $\{0,1\}^{d'}$) can be much larger than the one in the complete space (i.e., $\{0,1\}^d$).

As we will see in the next section, there exists a filter that almost match the asymptotic lower bound 
if $c\geq 2$. 
However,  if $1<c<2$ and $\varepsilon$ is sufficiently small, the upper bound has a $\BO{1/(c-1)^2}$ overhead:
although the upper bound is not optimal, the next theorem shows that a $1/(c-1)$ overhead is unavoidable when $1<c<2$.
To help in assessing the hypothesis in the theorem, 
we notice that, when $c=1+\frac{1}{\sqrt{r}}$, the theorem holds for $n\leq 2^{\BT{r}}$, $\varepsilon\leq 2^{-\BT{r}}$, $d=2^{\BOM{\sqrt{r}}}$
and it gives a $\BOM{nr^{3/2}}$ bound, whereas the previous theorem only gave $\BOM{nr}$.
We note that the next theorem can be integrated with the previous Theorem~\ref{thm:wclb1} to get an additive $nr/c$ or $nr^2/d$ more (according to the parameters).

\begin{theorem}\label{thm:wclb2}
Let $c\leq 2$, $\varepsilon\leq (c-1)/n$ be such that $ d(c-1)\geq  ((c-1)/\varepsilon)^{6/(r(c-1))} + (r(c-1))^3$.
Consider an $(r,c,\varepsilon)$-distance sensitive approximate membership filter with point-wise error $\varepsilon $ on a $S$ set of $n$ points in $\{0,1\}^d$.
Then, the filter requires  
$$\BOM{\frac{n}{c-1}\log\left(\frac{1}{\varepsilon}\right)}$$ 
bits in the worst case.
\end{theorem}

\begin{proof}
  The main idea of the proof is to use the optimal filter  in a one-way randomized protocol between two players (Alice and Bob) to send an arbitrary element $x$ of a given set $S$ from Alice to Bob who must identify which element he has:
  It is known (See the indexing problem~\cite{KushilevitzN97}) that such a protocol requires $\BOM{\log |S|}$ bits if the protocol succeeds with probability at least $2/3$ and the two players share random bits.
  The proof uses two families of error correcting codes, $\mathcal{C}$ and $\mathcal{M}$, that are explained below.
  Without loss of generality we assume that they are known to both Alice and Bob (the code families can be constructed with a deterministic brute-force algorithm).

Let $k=1/(c-1)$.
The error correcting binary code  $\mathcal C$ has $m=1/(n \varepsilon k)$ codewords, each one with length  $d_{\mathcal{C}}=d/k$ bits, weight $w=r/k$ and minimum Hamming distance between 
two codewords at least $\delta=r/k$.
\cite[Theorem 6]{GrahamS80} shows that such a code exists of size at least
\begin{eqnarray*}
\frac{d_{\mathcal{C}}^{w-\delta/2+1}}{\delta!}& \geq & 
\frac{(d(c-1))^{r(c-1)/2}}{(r(c-1))^{r(c-1)}}\\
& \geq & (d(c-1))^{r(c-1)/6}\\
& \geq & \frac{c-1}{\varepsilon}
\end{eqnarray*}
where in the third inequality we exploit the fact that $d(c-1)\geq (r(c-1))^3$ and in the last step we use $d(c-1)\geq ((c-1)/\varepsilon)^{6/(r(c-1))}$.

The error correcting binary code  $\mathcal M$ has $n$ codewords and minimum Hamming distance $rc$ (there is no requirement on codewords weights); we let $\mathcal M =\{m_1,\ldots, m_{n}\}$.
By the Gilbert-Varshamov~\cite{MacKay02} bound such a code $\mathcal M$ exists with length $d_{\mathcal{M}}= rc+\log n$. 

Alice arbitrary selects  $n$ codes $x_i=(x_{i,1}, \ldots,  x_{i,k-1})$ from the set $\mathcal C^k$.
Then, she encodes each $x_i$ into $\hat x_i= x_{i,1}\cdot \ldots \cdot x_{i,k} \cdot z_0 \cdot m_i $, where $\cdot$ denotes the concatenation of binary sequences, $z_0$ is a sequence of $r/k=r(c-1)$ zeros, and $m_i\in \mathcal M$. 
The length of each $\hat{x}_i$
is $d_x=k d_{\mathcal{C}}+d_{\mathcal{M}}+r/k= d+\log n+r(2c-1)$.
Finally, Alice inserts $\hat  x_0,\ldots, \hat  x_{n-1}$ into the optimal filter and sends the filter to Bob using $S(n,d_X, c, r)$ bits.

We now show  that Bob can reconstruct  each codeword $x_{i}$ by querying the filter at most $1/\varepsilon$ times.  
Codeword $x_{i,1}$ is obtained by performing a query with 
$q=q' \cdot z_2 \cdot z_3 \cdot m_i$ for every possible codeword $q'\in \mathcal{C}$, where $z_2$ is a sequence of $(k-1)\delta = (k-1)r(c-1) $ zeros, $z_3$ is a sequence of $r/k$ ones, and $m_i\in \mathcal M$. The distance between $q'$ and any $\hat x_j$ in the filter is
$D(\hat x_j,q)=D(x_{j,1},q')+D(x_{j,2}\cdot\ldots\cdot  x_{i,k}, z_2)+D(z_0,z_3)+D(m_j,m_i)$.
It holds that: 
\begin{enumerate}
\item  $D(x_{i,1},q')\geq r(c-1)$ if $q\neq x_{i,1}$ and $0$ otherwise; 
\item $D(x_{j,2}\cdot\ldots x_{i,k}, z_1) = (k-1) r(c-1)=r-r(c-1)$ since each codeword in $\mathcal{C}$ has weight $r(c-1)$;
\item  $D(z_0,z_3) = r(c-1)$;
\item  $D(m_{j},m_i)\geq r c$ if $m_j\neq m_i$
 and $0$ otherwise. 
\end{enumerate}
Therefore, $D(\hat x_j,q) = r$ if $x_{i,1}=q'$ and $m_i=m_j$, and $D(\hat x_j,q) \geq rc$ otherwise. 
A similar procedure holds for computing  $x_{i,j}$ for each $i$ and $j$.

Bob performs  $mk$  queries per $x_i$ and $nkm=1/\varepsilon$ queries in total.
The expected number of wrong queries is then $1$
and, if the protocol is repeated independently,  there  is a constant probability that all queries succeed. 
Since Bob is able to reconstruct an entry from the set $\mathcal S=\mathcal{C}^{nk}$, by the aforementioned result in~\cite{KushilevitzN97}, we have 
\begin{eqnarray*}
S(n,d_x, c, r, \varepsilon)&\geq& \BOM{\log \mathcal{S}}\\ 
&\geq & \BOM{ \log |\mathcal{C}|^{nk}}\\
&\geq & \frac{n}{c-1} \log (1/\varepsilon).\end{eqnarray*}
\end{proof}


%% file: wc_upper.tex
In this section we propose distance sensitive approximate membership filters with point-wise and average errors. 
We start in Section~\ref{sec:vect-sign-meth} by introducing the concept of vector signature.
It can be seen  as a succinct version of {\sc CountSketch}~\cite{CharikarCF04}, where we have thrown away information  not required for answering distance sensitive approximate membership queries.
In Sections~\ref{sec:filter-with-wc} and~\ref{sec:filter-with-average}, we then show how to use vector signatures to derive almost-optimal approximate membership filters with  point-wise and average errors respectively.

\subsection{Vector signatures}
\label{sec:vect-sign-meth}

\renewcommand{\d}{\kappa}
\newcommand{\cd}{c_\text{div}}
\newcommand{\cm}{c_\text{mod}}


A \emph{vector signature} is a suitable function mapping a vector from $\{0,1\}^d$ into 
$\BO{\frac{r}{(c-1)}+\left(\frac{c}{c-1}\right)^2\log\left(\frac{1}{\varepsilon}\right)}$ bits. 
The key feature of  vector signatures is that a suitable function of the signatures of two vectors $x$ and $y$ is smaller than or equal to a certain threshold  $\Psi$ if $D(x,y)\leq r$, while it is larger than $\Psi$ with probability $1-\varepsilon$ if $D(x,y)\geq cr$, as formalized in Theorem \ref{th:mainprop}.

\subsubsection*{Signature construction.}
The construction of the signature uses four parameters $m, \cm, \cd$ and $\delta$ that all depend on $r$, $c$ and $\varepsilon$.
Their values will be provided later.

Let $M$ be a $m\times d$ random matrix with entries chosen as follows.
For each $i\in\{1,\ldots,m\},j\in\{1,\ldots,d\}$, let $M_{i,j}$ denote the element in the $i$th row and $j$th column of $M$, and let $m_{i}$ denote the $i$th row.
Every entry of $M$ is initially set to $0$. 
Then each column $j$ of $M$ is constructed by performing $\delta=\BO{1+ (c/r)\log(1/\varepsilon) }$ updates, where each update is defined by the following three steps:
\begin{enumerate}
\item Select $s$ independently and uniformly from $\{-1, 1\}$.
\item Select a row $i$ uniformly at random from $\{1,\ldots,m\}$.
\item Update the entry at $M_{i,j}$ by adding $s$.
\end{enumerate}
We let $u_i$ denote the number of updates performed on all entries of row $m_{i}$; we have that $\|m_{i}\|_1\leq u_i$ (equality may not hold since two updates can affect the same entry and cancel each other).

For notational simplicity, we introduce the $\modl$ operator: it  is similar to the standard modulo operator, but it maps into the range $\left[ -\lfloor \cm/2\rfloor, \lceil \cm/2 \rceil \right)$ (the range is symmetric around zero when $\cm$ is even). Specifically,
\[\alpha {\modl} \cm = \left (\left(\alpha+\left\lfloor \frac{\cm}{2} \right\rfloor \right) \hspace{-.8em}\mod \cm\right) -\left\lfloor \frac{\cm}{2} \right\rfloor,\]
where $\bmod$ denotes the standard modulo operation into $[0, \cm)$.  

Let $\cd,\cm$ be suitable values with asymptotic value $\BO{c}$. The \emph{signature} of a vector $x\in \{0,1\}^d$ is then the $m$-dimensional vector $\sigma(x)$ defined by
\[
\sigma(x)_i = \left\lfloor \frac{(Mx)_i \modl \cm}{\cd}  \right\rfloor.
\]
Intuitively, the signature is a {\sc CountSketch} where we remove large values with $\modl \cm$, and remove the less significant bits with the division by $\cd$. 

The \emph{gap vector} between vectors $x$ and $y$ is the $m$-dimensional vector $\Gamma(x,y)$ where the $i$th entry is
\[
\Gamma(x,y)_i = \cd \left(\sigma(x)_i-\sigma(y)_i  \modl \cm \right).
\]
Finally, we refer to $\gamma(x,y)=\|\Gamma(x,y)\|_1$ as  the \emph{gap} between $x$ and $y$.

The following theorem describes the main property of signature vectors.
\begin{theorem}\label{th:mainprop}
Let $m=\BO{\frac{r}{(c-1)}+\left(\frac{c}{c-1}\right)^2\log\left(\frac{1}{\varepsilon}\right)}$, $\delta=\BO{1+ \tfrac{c}{r}\log(1/\varepsilon) }$, $\cd=\BO{c}$, and $\cm=\BO{c}$ be suitable values.
Then, there exists a value $\Psi=\BO{\delta r}$, such that for each pair of vectors $x,y \in \{0,1\}^d$:
\begin{itemize}
\item if $D(x,y)\leq r$, then $\gamma(x,y) \leq \Psi$;
\item if $D(x,y)> cr$, then $\gamma(x,y) > \Psi$ with probability at least $1-\varepsilon$.
\end{itemize}
\end{theorem}

We split the  proof of Theorem~\ref{th:mainprop} into two cases depending on the value of the approximation factor $c$: we first target constant approximation factors, and then we focus on larger values.
In the following proofs, we assume for notational convenience that two given vectors $x$ and $y$ differ on the first $D(x,y)$ positions.
We let $x'$ and $y'$ denote the prefix of length $D(x,y)$ of $x$ and $y$ (i.e., the positions where they differ), $M'$ denote the first $D(x,y)$ columns of $M$, $m'_i$ the $i$th row of $M'$, and $u'_i$ the number of updates affecting $m'_{i}$.

\subsubsection*{Proof of Theorem~\ref{th:mainprop} with $\mathbf{c=O(1)}$.}
For the case $c=\BO{1}$, we set the following parameters:
\begin{align*}
&m=\left\lceil 24 \frac{c^2}{c-1} \max\left\{ r, \frac{2}{c-1}\log\left(\frac{1}{\varepsilon}\right) \right\}\right\rceil,\\
&\cd = 1,\\
&\cm =2,\\
&\delta =1,\\
&\Psi =r.
\end{align*}
Note that the above values are consistent with the asymptotic values stated in Theorem~\ref{th:mainprop} since $c=\BO{1}$.
With these values, the signature definition simplifies to 
$$\sigma(x)_i = (Mx)_i \modl 2,$$ where each column of $M$ is a random vector with exactly one entry in $\{-1,1\}$ and the remaining $m-1$ entries set to zero.  
Then, the gap vector becomes:
$$\Gamma(x,y)_i=M(x-y)_i \modl 2 = M'(x'-y')_i \modl 2.$$ 
The first equality is true because  there is no rounding if $\cd=1$, and $\sigma$ is a linear function of $x$ and $y$.
The second one follows since the bit positions where $x$ and $y$ are equal do not affect the gap vector.

When $D(x,y)\leq r$, $M'$ contains at most $r$
entries in $\{-1,1\}$ and hence $\gamma(x,y)= \|M'(x'-y')\|_1\leq r$, proving the first part of Theorem~\ref{th:mainprop}.

Consider now the case $D(x,y)\geq cr$. The second part of Theorem~\ref{th:mainprop} follows by  two claims:
\begin{enumerate}[leftmargin=*,label=\emph{Claim \arabic*:}]
\item With probability at least $1-\varepsilon$, there are more than $r$ rows of $M'$ affected by an odd number of updates; we refer to these rows as \emph{odd rows}.
\item If $m'_{i}$ is an odd row, then $|\Gamma(x,y)_i|=1$.
\end{enumerate}
The two claims imply that $\gamma(x,y)=\sum_{i=1}^{m} |\Gamma_i(x,y)| > r=\Psi$ and hence Theorem~\ref{th:mainprop} follows.
The following Lemmas~\ref{lem:claim11} and~\ref{lem:claim12} show that the above claims hold.

\begin{lemma}[Claim 1]
  \label{lem:claim11}
Let $x,y$ be two input vectors in $\{0,1\}^d$, and let $M'$ be the sub-matrix of $M$ associated with the positions where $x$ and $y$ differ.
If $x$ and $y$ have distance at least $cr$, then there are more than $r$ odd rows in $M'$ with probability at least $1-\varepsilon$.
\end{lemma}
\begin{proof}
Consider the $D(x,y)$ updates used in the construction of $M'$.
If after the first $D(x,y)-cr$ updates there are more than  $(c+1)r$ rows with an odd number of updates, then the theorem follows: indeed, the remaining $cr$ updates can decrease the number of odd rows by at most $cr$.

Suppose now that there are  $Y_o\leq (c+1)r$ odd rows after the first $D(x,y)-cr$ updates, and consider  the last $cr$ updates.
Let  $Y_j$, with $j\in\{1,\ldots cr\}$ be a random variable set to 1 if the $j$th update affects an odd row, which then becomes an even row; $Y_i$ is set to 0 otherwise.
The probability that $Y_j=1$ is $p\leq (Y_o+j-1)/m\leq 3cr/m$ since there can be at most $Y_o+j-1$ odd rows before the $j$th update: the initial $Y_o$ odd rows and the rows affected by the previous $j-1$ updates. 
Let $Y=\sum_{j=1}^{cr} Y_j$. 
The expected value of $Y$ is $\mu=pcr\leq 3(cr)^2/m$.
Let $\eta=(c-1)r/(2\mu)-1$ (note that $\eta\geq 0$). By a Chernoff bound, we have
\begin{align*}
\Pr[Y\geq (c-1)r/2]=&\Pr[Y\geq \mu(1+\eta)]
\leq  e^{-\eta^2 \mu /2} \\
\leq & e^{-\left(\left(\frac{c-1}{c}\right)^2\frac{m}{24}+\frac{3(cr)^2}{2m}-\frac{(c-1)r}{2}\right)}\\
\leq & e^{-\left(\left(\frac{c-1}{c}\right)^2\frac{m}{24}-\frac{(c-1)r}{2}\right)}\\
\leq & \varepsilon.
\end{align*}

Therefore, with probability at least $1-\varepsilon$, there are  $Y< (c-1)r/2$ updates that affect odd rows and make them even. 
It follows that  the number of odd rows after all updates is then $Y_0+(cr-Y)-Y\geq  cr-2Y> r$.
\end{proof}

\begin{lemma}[Claim 2]
  \label{lem:claim12}
If row $m'_{i}$ is odd, then $|\Gamma_i(x,y)|= 1$.  
\end{lemma}
\begin{proof}
When $\delta=1$, there is one update per column and the number of non zero entries in $m'_{i}$ coincides with the number of updates (this may not happen if $\delta>1$).
Let $h_1,\ldots,h_{u_i}$ denote  the $u_i$ non zero entries in $m'_{i}$.
We have that $m_{i}(x'-y')=\sum_{j=1}^{u_i} M'_{i,h_j} (x'_{h_j}-y'_{h_j})$. Since $(x'_{h_j}-y'_{h_j})$ and $M'_{i,j}$  are in $\{-1,1\}$ and since $u_i$ is odd, then the sum must be odd and  $|\Gamma_i(x,y)|=|m'_{i}(x'-y') \modl 2|= 1$. 
\end{proof}

\subsubsection*{Proof of Theorem~\ref{th:mainprop} for $\mathbf{c=1+\Omega(1)}$.}
Let $\beta=15/(p_1 p_2)^2$ where $p_1$ and $p_2$ are suitable constants ($p_1 =  0.9$, $p_2 = 0.094$).
The proof presented here holds for $c\geq \sqrt{5 \beta/(4p_2^2)}\approx 545$. We believe that a smaller approximation factor $c$ can be obtained with a more careful analysis of the constants.
The parameters used in the signature construction are set as follows: 
\begin{align*}
&m=\left\lceil \beta \max\left\{\frac{r}{c}, \log\left(\frac{2}{\varepsilon}\right)\right\}\right\rceil,\\ 
&\cd=  \frac{2c}{\sqrt{5}\beta},\\
&\cm = 8 c,\\
&\delta=\left\lceil \frac{c}{r} \log\left(\frac{2}{\varepsilon}\right) \right\rceil,\\ 
&\Psi = \delta r + \max\left\{r,c \log\left(\frac{2}{\varepsilon}\right)\right\}.
\end{align*}
Note that the above values are consistent with the asymptotic values stated in Theorem~\ref{th:mainprop} since $c=1+\Omega(1)$.
In contrast to the $c=\BO{1}$ case, the gap vector and the gap cannot be expressed as a function of only the positions where $x$ and $y$ differ (i.e., $x'$ and $y'$). 
In fact, due to the division by $\cd$ and the floor operation, the gap vector may  depend on the positions where $x$ and $y$ coincide.
However, we can still provide upper and lower bounds on the gap that depend only on $x'$ and $y'$.
Indeed, it holds that:
\begin{align}\label{eq:approx_gamma}
\begin{split}
  |\Gamma_i(x,y)| &>  |m'_{i} (x'-y') \modl \cm | - \cd  \\
   |\Gamma_i(x,y)| &< |m'_{i} (x'-y') \modl \cm| +\cd.
\end{split}
\end{align}
Suppose $D(x,y)\leq r$, then by (\ref{eq:approx_gamma}) the gap can be upper bounded as follows:
\begin{align*}
\gamma(x,y)&=\sum_{i=1}^{m} |\Gamma_i(x,y)|\\
& \leq \sum_{i=1}^{m} \left(|m'_{i} (x'-y')\modl \cm   |+\cd\right)\\
& \leq  \cd m+\sum_{i=1}^{m} |m'_{i} (x'-y')|\\
& \leq \max\left\{r,c \log\left(\frac{2}{\varepsilon}\right)\right\} + \delta r=\Psi.
\end{align*}
In the third step, it is crucial to use $\modl $ instead of $\mod $ since it guarantees that $| \alpha \modl \cm|\leq |\alpha|$.
The last step is true since entries in $x'-y'$ are in $\{-1,1\}$ and $M'$ contains at most $\delta r$ non-zero entries.
The first part of Theorem~\ref{th:mainprop} follows.

Suppose now that $D(x,y)\geq cr$. 
We say that row $m'_{i}$ is dense if the number of updates $u_i$ is at least $4\delta D(x,y)/(5m)$.
The proof that the gap is larger than $\Psi$ with probability at least $1-\varepsilon$ relies on the following claims:
\begin{enumerate}[leftmargin=*,label=\emph{Claim \arabic*:}]
\setcounter{enumi}{2}
\item With probability at least $1-\varepsilon/2$, the number of dense rows is at least $p_1 m$.
\item With probability at least $p_2$, we have $|\Gamma_i(x,y)|> 2c/\sqrt{5\beta}$ for a dense row $m'_i$.  
\item With probability at least $1-\varepsilon$, there are at least $0.89 p_1 p_2 m$  rows such that $|\Gamma_i(x,y)|> 2c/\sqrt{5 \beta}$.
\end{enumerate}
Then, we have that $\gamma(x,y)=\sum_{i=1}^{m} |\Gamma_i(x,y)| > 0.89  p_1 p_2 m 2c/\sqrt{5\beta}  > 3 \max\left\{r,c \log\left(\frac{2}{\varepsilon}\right)\right\}>\Psi$ since $m=\lceil\beta \max\{r/c, \log(2/\varepsilon)\}\rceil$ and $\beta=15/(p_1 p_2)^2$. Thus, the second part of Theorem~\ref{th:mainprop}  follows.

Before proving the claims in Lemmas~\ref{lem:claim1b}-\ref{lem:claim3b}, we introduce three technical lemmas.
Lemma~\ref{lem:distr} gives a load bound on a balls and bins problem by using the bounded differences method to manage dependent random variables.
Lemma~\ref{lem:mod} bounds the probability of a sum of $\{-1,1\}$ random variables to be in a specified interval after a modular operation.
Finally, Lemma~\ref{lem:sum} gives a lower bound on the tail distribution of the sum of $\{-1,1\}$ random variables by leveraging the Berry-Esseen theorem.

\begin{lemma}\label{lem:distr}
  Consider $p$ balls thrown uniformly and independently at random into $q$ bins, with $p\geq q$.
  For every $\alpha>0$ with probability at least $1-\varepsilon$, there are more than $q\left(1-e^{-\alpha}-\sqrt{ \log(1/\varepsilon)/(2 {q})}\right)$ bins with at least $\left(p/q\right)\left(1-\sqrt{2\alpha q/p}\right)$ balls.
\end{lemma}
\begin{proof}
For every $i\in\{1,\ldots, p\}$ and $j\in\{1,\ldots, q\}$, define the following random variable:
  \[X_{i,j}=\begin{cases}
    1 \text{ if ball $i$ landed in bin $j$}\\
    0 \text{ otherwise}
  \end{cases}
  \]
Let also $X_j=\sum_{i\in[p]} X_{i,j}$ be the number of balls in the $j$th bin; the expected value of $X_j$ is $\mu = p/q$ for each $j$.
  Since the balls are thrown independently a Chernoff bound gives:
  \[\Pr\left[X_j\leq \mu \left(1- \sqrt{2\alpha/\mu}\right)\right]  \leq e^{-\alpha}\]
Consider now the random variable $Y_j$:
\[Y_j=\begin{cases}
    1 \text{ if }X_j > \mu \left(1- \sqrt{2\alpha/\mu}\right)\\
    0 \text{ otherwise}
  \end{cases}
\]
Let $Y=\sum_{j=1}^{q} Y_j$; we use $Y_{Y_1,..,Y_{q}}$ to denote the actual value of $Y$ with the specified values.
Since there is dependency among the $Y_j$, we use the method of bounded differences~\cite{DubhashiP09} to bound the tail distribution, instead of a Chernoff bound.
The random variable $Y$ satisfies the Lipschitz property with constant $1$, that is:
\[|Y_{Y_1, \ldots, Y_i, \ldots,  Y_q} - Y_{Y_1,  \ldots, Y'_i, \ldots, Y_q}|  = |Y_i-Y'_i| \leq 1\] whenever $Y_i\neq Y'_i$  for every $i\in\{1,\ldots, q\}$.
By the method of bounded differences ~\cite[Corollary 5.2]{DubhashiP09}, 
we  get $\Pr\left[Y\leq \E[Y]- t\right] \leq e^{-2t^2/q}$, and
then  $\Pr\left[Y> \E[Y]- t\right] \geq 1-\varepsilon$ if $t= \sqrt{(q/2) \log(1/\varepsilon)}$.
Since 
$\E[Y]\geq q \left(1-\Pr\left[X_j\leq \mu \left(1- \sqrt{2\alpha/\mu}\right)\right]\right)\geq q\left(1-e^{-\alpha}\right),$ 
 the claim follows.
 \end{proof}
 
\begin{lemma}\label{lem:mod}
Consider a sequence $s_1,\ldots, s_k$ of independent and evenly distributed random variables  in $\{1,-1\}$, and an 
arbitrary  value $q\in \mathbb{N}$.
 Let $S=\sum_{i=1}^{k} s_i$ and $S_q = S \modl q$.
 Then for all values $a, b$ such that $0 \leq a < b \leq \lceil q/2\rceil $ and $b-a\geq q/3$, we have:
\begin{equation}
\label{eq:boundmod}
\frac{\Pr[|S|\geq a]}{2} < \Pr[a\leq |S_q| <b ]  < \Pr[|S|\geq a] .
\end{equation}
\end{lemma}
\begin{proof}
Let $k'= k/q$ and assume for the sake of simplicity that $k'$ is an integer, and that $q$, $b$ and $a$ are even (the proof extends to the general case with minor adjustments). We define the following four quantities:
\begin{eqnarray*}
H_1 \hspace{-1em}&= \sum_{\ell=0}^{k'-1} \hfill  &\Pr\left[\ell q + a \leq |S| < \ell q + b\right]; \\
 H_2 \hspace{-1em} & = \sum_{\ell=0}^{ k'-1} &   \Pr\left[\ell q + b \leq |S| \leq (\ell+1)q -b\right] ; \\
 H_3 \hspace{-1em} &= \sum_{\ell=0}^{ k'-1} & \Pr\left[( \ell+1) q -b < |S| \leq (\ell+1) q-a\right]; \\
H_4 \hspace{-1em} &= \sum_{\ell=0}^{ k'-1}  &\Pr\left[(\ell+1)q - a < |S| < (\ell+1)q +a\right].
\end{eqnarray*}
Standard computations show that:
$\Pr[a \leq |S_q| < b]= H_1+H_3$ and that $\Pr[|S| \geq a]=H_1+H_2+H_3+H_4$.
We then have that  $\Pr[a\leq |S_q| <b ] < \Pr[|S|\geq a]$, and the right side of the inequality in~(\ref{eq:boundmod}) follows.

We now focus on the other side of the inequality.
We  prove that $H_1\geq H_2+H_4$.
The random variable $S$ has value $i$, with $i\in [-k,k]$ if there are $(k+i)/2$ terms set to $+1$ and $(k-i)/2$ terms set to $-1$.
If $k+i$ is odd, this cannot happen and hence $\Pr[S=i]=0$.
On the other hand, if $k+i$ is even, we  have $\Pr[S=i] = \binom{k}{(k+i)/2} \frac{1}{2^k}$  since
the $s_i$ terms are independent and evenly distributed.
Note that $\Pr[S=i]$ is decreasing for the even values of~$i$.

Let us define ${\alpha \brack \beta/2}$ to  $\binom{\alpha}{\beta/2}$ if $\beta$ is even and to $0$ if $\beta$ is odd: we thus have $\Pr[S=i]={k \brack (k+i)/2}$ for any even/odd $i$.
Let $\beta\geq \alpha$ and $\gamma\geq 1$, we have the following property:
\begin{align*}
{\alpha \brack \beta/2} + 
{\alpha \brack (\beta+1)/2} >
{\alpha \brack (\beta´+\gamma)/2}+
{\alpha \brack (\beta´+\gamma+1)/2}.
\end{align*}
The correctness of the property follows from the fact that there is exactly one non zero term on each side of the inequality by definition of  ${\alpha \brack \beta/2}$, and the non zero one on the right is decreasing in $\gamma$.

We  then have, for any integer $\ell\geq 0$, that :
\begin{align*}
\Pr[a + \ell q &\leq |S| < b +  \ell q]  = 2\sum_{j=a + \ell q }^{ b +  \ell q-1}  {k \brack \frac{(k+j)}{2}} \frac{1}{2^k}\\
\geq 2& \sum_{j=a + \ell q }^{a+(\ell+1) q -2b}  {k \brack \frac{(k+j)}{2}} \frac{1}{2^k}\\ 
&+ 2 
\sum_{j=a+(\ell+1) q -2b+1}^{a+(\ell+1) q -2(b-a)-1}  {k \brack \frac{(k+j)}{2}} \frac{1}{2^k},
\end{align*}
where the  step follows by the initial assumption  $(b-a)\geq q/3$.
By using the  above property of ${\alpha \brack \beta/2}$, we shift the indexes of the above summations (we add $b-a$ to the first sum and $2(b-a)$ to the second one):
\begin{align*}
\hspace{-.3em}\Pr&[a + \ell q \leq |S| < b +  \ell q]  \\
\hspace{-.3em} > &  2\sum_{j=\ell q+b}^{(\ell +1) q-b} {k \brack \frac{(k+j)}{2}} \frac{1}{2^k} +  2\sum_{j=(\ell+1) q-a+1}^{(\ell+1) q+a-1} {k \brack \frac{(k+j)}{2}} \frac{1}{2^k} \\
\hspace{-.3em} \geq & \Pr[\ell q +b {\leq } |S| {\leq} (\ell+1)q-b] \\ & \hspace{2em}+\Pr[(\ell+1)q-a {<} S {<} (\ell+1)q +a]\\
\hspace{-.3em}  \geq & H_2+H_4. 
\end{align*}
(Note that the derivation requires some adjustments when $q$, $b$ or $a$ are not even).
Therefore,  $\Pr[|S|\geq a] = H_1+H_2+H_3+H_4 <  2(H_1+H_3) \leq 2 \Pr[a\leq |S_q| <b ] $. 
The left side of the inequality in~(\ref{eq:boundmod}) follows.
 \end{proof}

\begin{lemma}\label{lem:sum}
Let $S=\sum_{i=1}^{k} s_i$, where the $s_i$ terms are independent and unbiased random variables in $\{-1,+1\}$ and let $\alpha>0$ be any arbitrary value. 
Then, 
$$\Pr[|S|\geq  \alpha \sqrt{k}]\geq \frac{2\alpha}{\sqrt{2\pi}(\alpha^2+1)e^{\alpha^2/2}}-\frac{1}{2\sqrt{k}}.$$
\end{lemma}
\begin{proof}
We observe that $\E[s_i]=0$, $\sigma^2 = \E[s^2_i]=1$ and $\rho = \E[|s_i|^3]=1$.
By the Berry-Esseen theorem~\cite{Berry1941}, we have that the random variable $Q = S/(\sqrt{k}\sigma)=S/\sqrt{k}$ can be approximate by a standard normal distribution $\mathcal N(0,1)$ with error  
$$
| \Pr[Q\leq x] -\Psi(x) | \leq \frac{C \rho}{\sigma^3 \sqrt{k}},
$$ 
where $\Psi(x)$ is the cumulative distribution function of the standard normal distribution $\mathcal N(0,1)$ and $C$ is a suitable constant smaller than $1/2$~\cite{Tyurin10}.
The above inequality can be rewritten as
\[
| \Pr[Q> x] -\Psi^c(x) | \leq \frac{1}{2\sqrt{k}},
\] 
with $\Psi^c(t)=1-\Psi(x)$. 
We then get
\begin{align*}
\Pr[|S|\geq \alpha \sqrt{k} ] &= 2\Pr[S\geq \alpha \sqrt{k} ] \\ 
&= 2 \Pr[Q\geq \alpha]\\
&\geq 2\Psi^c(\alpha) - \frac{1}{2\sqrt{k}}.
\end{align*}
Since  $\Psi^c(x)\geq x/(\sqrt{2\pi} (x^2+1) e^{x^2/2})$~\cite{Cook09,Abramowitz74}, the lemma follows by inserting the bound for $\Psi^c(x)$.
\end{proof}

We are now ready to prove the three claims used in the proof of Theorem~\ref{th:mainprop} for $c=\BOM{1}$.

\begin{lemma}[Claim 3]\label{lem:claim1b}
With probability at least $1-\varepsilon/2$, the number of dense rows in $M'$ is at least $p_1 m$, with $p_1 = 0.9$.
\end{lemma}
\begin{proof}
Matrix $M'$ is obtained by performing $\delta$ random updates per column independently and uniformly distributed.
The number of updates $u_i$ affecting row $m'_{i}$ is distributed as the number of balls in a bin after randomly throwing $\delta D(x,y)$ balls into $m$ bins.
By applying Lemma~\ref{lem:distr} with $\alpha=3$, it follows that, with probability at least $1-\varepsilon/2$, there are more than
\[m'\geq (1-1/e^3-\sqrt{\log(2/\varepsilon)/(2m)}) m\geq p_1 m\]  rows where
\begin{align*}
u_i& \geq \frac{\delta D(x,y)}{m} \left(1-\sqrt{\frac{6 m}{\delta  D(x,y)}}\right)\\
   & \geq \frac{4\delta D(x,y)}{5m}
\end{align*} 
as soon as $c\geq 5 \sqrt{6\beta+1}$ (which is true under the initial hypothesis $c\geq \sqrt{5\beta/(4p_2^2)}$). These $m'$ rows are then dense.
\end{proof}

\begin{lemma}[Claim 4]\label{lem:claim2b}
If $m'_{i}$ is dense, then $|\Gamma_i(x,y)|> 2c/\sqrt{5\beta}$ with  probability at least $p_2=0.094$.
\end{lemma}
\begin{proof}  
Let $K=2c/\sqrt{5\beta}(1+1/\sqrt{\beta})$ and assume that the inequality  $|m'_{i} (x'-y') \modl \cm|\geq  K$ holds. 
Then, the lemma follows by applying~(\ref{eq:approx_gamma}):
\begin{align*}
|\Gamma_i(x,y)| & > |m'_{i} (x'-y') \modl \cm| - \cd \\
& \geq K - \cd\\
& = c/(\sqrt{5}\beta) + 2c/\sqrt{5\beta}- \cd\\
& = 2c/\sqrt{5\beta}.
\end{align*}

We now show that the above inequality holds (i.e., $|m'_{i} (x'-y') \modl \cm|\geq  K$).
The inner product $m'_{i} (x'-y')$ can be rewritten as $\sum_{j=1}^{u_i} \sigma_j (x'-y')_{f(j)}$, where $f(j)$ is the position in $m'_{i}$ affected by the $j$th update.
Since $(x'-y')$ has entries in $\{-1,1\}$ and the $\sigma_j$ are independent,  $m'_{i} (x'-y')$
has the same density function as $S=\sum_{j=1}^{u_i} \sigma_j$.
Then,
\begin{align*}
\Pr[|M'_i & (x'-y') \modl \cm|\geq  K]\\
  &=
\Pr[|S \modl \cm| \geq K] \\
& > \frac{\Pr[|S|\geq K]}{2},
\end{align*}
where the last step  follows by applying Lemma~\ref{lem:mod} with $a=K$, $b=\cm/2$ and $q=\cm$ (note that $b-a\geq \cm/3$).
To lower bound $\Pr[|S|\geq K]$, we apply Lemma~\ref{lem:sum} with $\alpha=1+1/\sqrt{\beta}$ since $K\leq \sqrt{u_i}(1+1/\sqrt{\beta})$. Hence,
\begin{align*}
\frac{\Pr[|S|\geq K]}{2}&\geq
\frac{\Pr[|S|\geq (1+1/\sqrt{\beta})\sqrt{u_i}]}{2}\\
&\hspace{-2em}\geq \frac{1+1/\sqrt{\beta}}{\sqrt{2\pi} ((1+1/\sqrt{\beta})^2+1) e^{(1+1/\sqrt{\beta})^2/2}}-\frac{1}{4\sqrt{u_i}}\\
&\hspace{-2em}\geq p_2,
\end{align*}
where the last step follows by observing that $\sqrt{u_i}\geq 2c/\sqrt{5\beta}  \geq 1/p_2$, and then by numerically evaluate the resulting bound.
\end{proof}

\begin{lemma}[Claim 5] \label{lem:claim3b}
With probability at least $1-\varepsilon$, there are at least $0.89 p_1 p_2 m$ rows such that $|\Gamma_i(x,y)|> 2c/\sqrt{5\beta}$.
\end{lemma} 
\begin{proof}
By Lemma~\ref{lem:claim1b}, there are $m'\geq p_1 m$ dense rows with probability $1-\varepsilon/2$. 
For each dense row, let $Y_i$ be a random variable sets to 1 if $|\Gamma_i(x,y)|> c/\sqrt{\beta}$, and 0 otherwise.
By the previous Lemma~\ref{lem:claim2b}, we have that $\Pr[Y_i=1]\geq p_2$.
Let $Y=\sum_{i=1}^{m'} Y_i$.
Since the $Y_i$ are independent and $\E[Y]=p_2 m'$, a Chernoff bound gives:
\[
  \Pr\left[Y< p_2 m'\left(1- \sqrt{2\log(2/\varepsilon)/(p_2 m')}\right)\right]\leq \varepsilon/2.
\]
By plugging in the actual values of variables, 
we have $ \Pr[Y< 0.89 p_1 p_2 m]\leq \varepsilon/2$.

Therefore, by an union bound there are at least $ p_1 m$ dense rows and 
at least $0.89 p_1 p_2 m$ of them satisfy
 $|\Gamma_i(x,y)|> c/\sqrt{\beta}$.
\end{proof}

\subsection{A filter with point-wise error}\label{sec:filter-with-wc}
A distance sensitive approximate membership filter with point-wise error is obtained by just storing the $n$ signatures of the points in $S$. 
We have the following theorem:

\begin{theorem}
There exists a $(r,c,\varepsilon)$-distance sensitive approximate membership filter with point-wise error which requires 
$$
\BO{n\left(\frac{r}{(c-1)}+\left(\frac{c}{c-1}\right)^2\log\left(\frac{n}{\varepsilon}\right)\right)}
$$
bits for any $c>1$ on  a set $S$ of $n$ points.
When $c\geq 2$, the filter uses 
$\BO{n\left(\frac{r}{c} + \log\left(\frac{n}{\varepsilon}\right)}\right)$ bits, and it is optimal if  $r/c\geq \log(n/\varepsilon)$ or $\varepsilon\leq 1/n^{1+o(1)}$.
\end{theorem}
\begin{proof}
We assume to have a shared source of randomness that can be used to recover the random matrix $M$ without storing it.
Consider the $n$ signatures of points in $S$ constructed with error $\varepsilon'=\varepsilon/n$. 
By an union bound, the $n$ signatures give a false positive with probability $\varepsilon$. 
Since each signature requires 
$
\BO{\frac{r}{(c-1)}+\left(\frac{c}{c-1}\right)^2\log\left(\frac{n}{\varepsilon}\right)}
$
bits by Theorem~\ref{th:mainprop}, the first part of the claim follows.
The optimality with $c\geq 2$ of the filter follows from Theorem~\ref{thm:wclb1}.
\end{proof}

\subsection{A filter with average error}
\label{sec:filter-with-average}
The point-wise error filters are of course  valid average error filters, but in this setting we can also construct space efficient filters with a $c=1$ approximation factor. 
Define $Q_{\text{$r$-far}}=\{x\in\{0,1\}^d\; | \; D(x,S)\geq r\}$ and similarly $Q_{\text{$(r;cr)$-far}}=\{x\in\{0,1\}^d\; | \; r\leq D(x,S)\leq cr\}$.

By setting $c=r$ in the point-wise filter, we obtain an average error filter with $c=1$ which matches the $\BOM{n\log(1/\varepsilon)}$ lower bound of Theorem~\ref{avg_error_thm} for small $r$. 
Interestingly, this space bound shows that it is possible to support distance sensitive membership queries in the average error setting with the asymptotic space bound of a Bloom filter.
\begin{theorem}
Let $r \leq \sqrt{d}$, $n\leq 2^{d/3}$ and $\varepsilon\geq 1/2^{d-2}$. Then, there exists an optimal $(r,1,\varepsilon)$-distance sensitive approximate membership filter with average error which requires  $\BO{n \log(1/\varepsilon)}$ bits on  a set $S$ of $n$ points.
\end{theorem}
\begin{proof}
Let us consider a $(r,r,\varepsilon/4)$-filter $\mathcal F$ with point-wise guarantees.
The amount of false positives accepted by $\mathcal F$ is $P \leq (\varepsilon/4) |Q_\text{$r^2$-far}| + |Q_\text{$(r;r^2)$-far}|$.
We have $|Q_\text{$(r;r^2)$-far}|\leq n r^2 \binom{d}{r^2} \leq (\varepsilon/4) 2^d$ since
$d \geq r^2$, $n\leq 2^{d/3}$ and $\varepsilon \geq 4/2^{d/2}$.
Trivially, we also have that $|Q_\text{$r^2$-far}|\leq 2^d$.
We see that $P\leq  \varepsilon 2^{d-1}$.  

Now note that $|Q_\text{$r$-far}|\geq 2^d-n r \binom{d}{r}\geq 2^{d-1}$ by $d \geq r^2$ and $n\leq 2^{d/3}$.

We combine the two bounds to see $P\leq \varepsilon 2^{d-1}\leq \varepsilon|Q_\text{$r$-far}|$. 
The optimality of $\mathcal{F}$ follows from Theorem~\ref{avg_error_thm} since $r^2/d<1$ and $n \log(1/\varepsilon)$ is a lower bound.
\end{proof}


%% file: conclusion.tex
To the best of our knowledge, this paper is the first that presents and gives upper and lower space bounds for the problem of distance sensitive filters without false negatives.
We have introduced a distance sensitive signature for Hamming vectors and shown that it can be used to derive filters with point-wise and average errors. The proposed filters are optimal under certain assumptions, but it is an open question to close the gap without these assumptions, specifically when $\varepsilon$ is large.

Another interesting research direction is to investigate trade-offs between space and query time: our filter requires reading all signatures at query time and it is not clear to which extent the query time can be improved.
We finally remark that, although the constants in the asymptotic analysis of our filters are large, a preliminary experimental analysis shows that the signatures exhibit significant space savings and are easy to implement.